%% file: farming.tex
\documentclass[12pt]{article}%

\input{prefix.tex}

\begin{document}

\title{On the Expected Complexity of\\ Voronoi Diagrams on Terrains}

\author{%
   Anne Driemel%
   \AnneThanks{This work has been supported by the Netherlands
      Organisation for Scientific Research (NWO) under RIMGA
      (Realistic Input Models for Geographic Applications).}%
   \and%
   Sariel Har-Peled%
   \SarielThanks{Work on this paper was partially supported by a NSF
      AF award CCF-0915984.}%
   \and %
   Benjamin Raichel%
   \BenThanks{} }

\date{\today}

\maketitle

\begin{abstract}
    We investigate the combinatorial complexity of geodesic Voronoi
    diagrams on polyhedral terrains using a probabilistic
    analysis. Aronov \etal \cite{abt-cbvdrt-08} prove that, if one
    makes certain realistic input assumptions on the terrain, this
    complexity is $\Theta(n + m \sqrt{n})$ in the worst case, where
    $n$ denotes the number of triangles that define the terrain and
    $m$ denotes the number of Voronoi sites. We prove that under a
    relaxed set of assumptions the Voronoi diagram has expected
    complexity $O(n+m)$, given that the sites have a uniform
    distribution on the domain of the terrain (or the surface of the
    terrain).  Furthermore, we present a worst-case construction of a
    terrain which implies a lower bound of $\Omega(n m^{2/3})$ on the
    expected worst-case complexity if these assumptions on the terrain
    are dropped.  As an additional result, we can show that the
    expected fatness of a cell in a random planar Voronoi diagram is
    bounded by a constant.
\end{abstract}

\section{Introduction}
\seclab{intro}

Voronoi diagrams on terrains are a basic geometric data-structure that
have a variety of applications in many areas: geographic information
science (\textsc{\si{gis}}) \cite{ajt-vbra-01, dg-tsvdgn-08,
   pa-cnmwvd-06, cg-frn-10}, robot motion planning \cite{ts-mpgvd-89},
mesh generation \cite{kwr-gvdps-97} and image analysis
\cite{sf-clssf-04, wdy-paadm-08} to name a few.  The geodesic Voronoi
diagram of point sites on a polyhedral terrain is a subdivision of the
surface into cells, according to the set of sites, such that every
cell contains exactly the surface points which are closest to the site
that is associated with the cell. Here, the distance is measured by
the length of the shortest path on the terrain. It is tempting to
believe that in practice -- that is, given that the terrain is
well-behaved -- the complexity of such a geodesic Voronoi diagram
should be linear, because of its similarity to the Euclidean Voronoi
diagram of point sites in the plane.

However, in the worst case, this complexity can be much higher, even
if one makes certain realistic assumptions on the shape of the
terrain.  Indeed, Aronov \etal \cite{abt-cbvdrt-08} show that the
worst-case complexity is $\Theta(n+m\sqrt{n})$ for a certain class of
well-behaved terrains, where $n$ is the number triangles that define
the terrain and $m$ is the number of Voronoi sites. This shows that
assuming realistic input indeed brings the complexity down, but it is
still far from being linear.  They conjecture that, in order to prove
a linear bound, one needs to make further assumptions on how the sites
are distributed.  Our purpose in this paper is to study the complexity
of a geodesic Voronoi diagram if we assume that the sites are being
chosen randomly from the terrain.

\paragraph{Previous work.}
Analyzing the expected complexity of geometric structures for random
inputs has a long history in computational geometry. See for instance
the work of \Renyi and Sulanke \cite{rs-udkhv-63} and Raynaud
\cite{r-slcdn-70} on the complexity of convex hulls of random
points. Weil and Wieacker give an overview of related results in
\cite{ww-sg-93}.  Naturally, Voronoi diagrams (and their counterpart,
Delaunay triangulations) have been analyzed in this way as they
are a fundamental data-structure, used in fields such as mesh
generation \cite{r-draqt-95}, surface reconstruction \cite{d-csra-11},
molecular biology, and many others. It is well-known that the size of
the Voronoi diagram of point sites in $\Re^3$ is quadratic in the
worst-case, however it is near-linear in most practical situations. To
address this dichotomy, people have investigated the complexity when
the point sites are: (i) generated by a random processes, (ii) well
spaced, (iii) have bounded spread, or (iv) were sampled from surfaces
according to curvature. See \cite{deg-eeg-08} and references therein
for more information on such work.  In particular, there is a vast
amount of work on \emph{Poisson Voronoi Diagrams}
(\emph{\PVD{}}). Here, the domain has a density associated with it
(say, the area). The probability of $n$ points to appear in an area of
measure $\mu$ has, as the name suggests, a Poisson distribution
parameterized by the area. Similarly, the distribution of points
selected into disjoint areas is independent. Poisson Voronoi diagrams
are used in many areas, such as physics, biology, animal ecology, and
others. See \cite{obsc-stcav-00, jn-osdpv-04} and references
therein. However, this work does not seem to have considered geodesics
at all.

In this paper, we are interested in the complexity of geodesic Voronoi
diagrams on polyhedral terrains.  Moet \etal \cite{mks-ort-08} were
the first to study this complexity using a set of parameterized
assumptions that describe realistic terrains. In this approach, one
assumes that a certain property, for example, the maximum slope of the
terrain, can be bounded by a constant independent of the input size.
This allows one to avoid certain worst-case configurations which are
highly unlikely to occur in practice.  Instead, the analysis is
confined to classes of well-behaved inputs and consequently this
method is described as using \emph{realistic input models}.  Moet also
did an experimental validation of the used parameters
\cite{m-ccvge-08} and confirmed that the parameters indeed behave like
constants on realistic terrains.  The realistic input models
introduced in this work have also been adopted by subsequent papers.
As such, Aronov \etal \cite{abt-cbvdrt-08} improved the bounds given
by Moet \etal and showed that
\begin{inparaenum}[(i)]
    \item the bisector between two sites has worst-case complexity
    $\Theta(n)$, (where $n$ denotes the number of triangles of the
    terrain) if the triangulation is low density and the lifted
    triangles have bounded slope; and
    \item that the worst-case combinatorial complexity of the Voronoi
    diagram is $\Theta(n+m\sqrt{n})$, (where $m$ denotes the number of
    sites) if in addition the triangles are of similar size and the
    aspect ratio of the domain is bounded.
\end{inparaenum}
The realistic assumptions made in these papers are described in more
detail in the next section.

Finally, note that Schreiber and Sharir \cite{ss-spcp-08} showed how
to compute an implicit representation of the geodesic Voronoi diagram
on the surface of a convex polyhedron, in time $O((n+m)\log (n+m))$,
so that the site closest to a query point can be reported in time
$O(\log (n+m))$. Schreiber \cite{s-sprp-07} also extended their method
for single-source shortest paths to the case of non-convex polyhedra
using several realistic input models.  Naturally, these analyses do
not inform about the complexity of the explicit Voronoi diagram.

\paragraph{Our results.}
We study the expected complexity of geodesic Voronoi diagrams on
terrains.  To this end, we use realistic input assumptions on the
terrain, and sample the sites uniformly at random from the domain of
the terrain. See \secref{prelims} for the exact definitions.  In
\secref{upper:bound} we show that under these assumptions the
complexity of the geodesic Voronoi diagram is indeed linear. That is,
we show that the complexity is bounded by $O(n+m)$, where $n$ is the
complexity of the terrain, and $m$ is the number of sites being
randomly picked.  The constants in the asymptotic analysis depend on
how well-behaved the terrain is, which is formalized using the input
models described in the next section. See \thmref{main} for the exact
result.  In \secref{lower:bound} we analyze the expected complexity if
these assumptions on the shape of the terrain are dropped. In
particular, in \thmref{lower:bound} we show a lower bound of
$\Omega\pth{nm^{2/3}}$.  This lower bound, in a sense, justifies the
input assumptions made previously, since it implies that the
randomness assumption by itself is not sufficient if we want the
geodesic Voronoi diagram to have a low complexity.  The construction
that leads to this lower bound is intricate and requires a careful
balancing of the variance of the distances of the sampled sites, and
how closely they can be packed together.  Furthermore, in
\apndref{fat:cells} we show that in expectation the Voronoi cells
generated by a uniform sample in the plane are in expectation fat;
that is, they are nicely behaved in some sense.


\paragraph{Organization.}
In \secref{prelims} we introduce the concepts used.
\secref{upper:bound} contains the poofs of the upper bound on the
complexity of the Voronoi diagram. The lower bound is proved in
\secref{lower:bound}. Finally, we conclude in \secref{conclusions}.
The expected fatness of a Voronoi cell is analyzed in
\apndref{fat:cells}.

\section{Preliminaries}
\seclab{prelims}

\subsection{Voronoi Diagrams on Terrains}

A polyhedral terrain $\Terrain$ is defined by a triangulation
$\Triangulation$ of $n$ vertices $\Vertices$ in $\Re^2$, a convex
domain $\Domain \subseteq \Re^2$ which contains $\Vertices$, and a
height function on these vertices.  The \emphi{surface} of the terrain
is defined by the triangles of $\Triangulation$ lifted according to
this height function. We refer to $\Terrain$ simply as a
\emphi{terrain} and we denote the set of edges of the triangulation
with $\Edges$. For simplicity of exposition we restrict our discussion
to the case where $\Domain$ is the unit square, however, our results
can be easily extended to the more general case of convex regions with
bounded aspect ratio.  For two points $\pntA,\pntB \in \Domain$, we
denote their euclidean distance in the $(x,y)$-plane with
$\distX{\pntA}{\pntB}$.  When $\pntA$ and $\pntB$ are lifted to the 
surface of $\Terrain$, we define their geodesic distance to be the length of 
the shortest path connecting them that is constrained to lie in the 
surface of $\Terrain$, and we denote this value by 
$\distGeo{\pntA}{\pntB}$.

The (geodesic) \emphi{Voronoi diagram} of a set of $m$ points on
$\Terrain$ (which are called \emphi{sites}) is a planar subdivision of
the surface of $\Terrain$, where every cell of the subdivision is
associated with exactly one site, and such that for any point in the
cell the associated site is the closest site, where the distances are
measured using the geodesic distance.  We denote
the Voronoi diagram with $\VD{\Sites}$, where $\Sites$ denotes the set
of sites, and we call a cell of the subdivision a \emphi{Voronoi
   cell}.  The \emphi{bisector} between two sites $\pntA$ and $\pntB$
on the surface of $\Terrain$ is defined as the set of points $\pnt$,
such that $\pnt$ has the same distance to $\pntA$ and $\pntB$.  The
Voronoi diagram can be represented as the structured set of curves and
straight lines which delineate the Voronoi cells and which are subsets
of the bisectors between these points.  We call a point which is
incident to at least three cells a \emphi{Voronoi vertex} and we call
each maximally connected subset of the bisector incident to two
Voronoi cells a \emphi{Voronoi edge} (note that two cells can have
multiple edges between them).  Usually one assumes general position of
the sites so that no two sites are equidistant from a terrain vertex,
which ensures that bisectors are one-dimensional and that the Voronoi
cells subdivide the terrain surface without overlap, see also
\cite{abt-cbvdrt-08}.  In our case the sites are randomly
sampled, and so general position is implied.

Since a terrain $\Terrain$ is defined by a height function over a
domain $\Domain$, there is a natural bijection between points of
$\Terrain$ and points of $\Domain$.  Hence, the various objects
defined in the previous paragraph can be viewed either in $\Terrain$
or in $\Domain$.  Generally in the paper we shall refer to these
objects by their projection in $\Domain$, unless otherwise stated.

\subsection{Input Model}

The main idea of \emphi{realistic input models} is to parametrize
certain properties of the input, which are suspected to capture
contrived configurations leading to high complexities or running
times. In cases where there exists a high discrepancy between the
theoretical bounds and the complexities observed in practice, it is
often useful to analyze the complexities not only as a function of the
input size, but also with respect to these parameters.  This sometimes
leads to more informative asymptotic bounds.  As such, the realistic
input assumptions do not only distinguish between ``good'' and ``bad''
input, instead they enable a more differentiated view on which inputs
are 'better' or 'worse'.

In this paper, we use the following realistic input model.  A set of
line segments is $\ldConst$-\emphi{low density} if and only if the
number of edges that intersect an arbitrary ball, which are longer
than the radius of the ball, is smaller than $\ldConst$.  Low density
has been used in the analysis of many different geometric problems,
see \cite{bksv-rimga-02} for an overview.

To model a \emphi{realistic terrain} we adopt the realistic
assumptions made in \cite{abt-cbvdrt-08}.  According to these
assumptions, there exist constants $\ldConst$ and $\slConst$
independent of $n$, such that
\begin{compactenum}[\qquad(i)]
    \item the set of edges of the triangulation, is a
    $\ldConst$-\emph{low density} set, and
    \item any line segment embedded in the lifted triangulation has
    slope at most $\slConst$.
\end{compactenum}

We now state some useful facts that follow from these assumptions.
For notational ease in the rest of the paper we define the constant
$\slFactor = \sqrt{1+\slConst^2}$.

First, the number of pairs of objects from two low density sets, which
intersect each other is linear in the total number of objects in these
sets.  This can be easily verified, the idea is to charge each
intersecting pair to the smaller of the two objects.

\begin{fact}[\cite{abt-cbvdrt-08}]
    Let $A$ be a set of $n$ objects with $\lambda$-low density and let
    $B$ be a set of $m$ objects with $\phi$-low density, then the
    number of pairs of objects $(u,v) \in A \times B$, such that $u$
    intersects $v$ is $O(\lambda m + \phi n)$.
    
    \factlab{l:d:intersections}
\end{fact}

Second, since the slope is bounded by a constant, the geodesic
distance is the same as the euclidean distance up to a constant
factor, and similarly, geodesic disks have an area that is
approximately the same as that of planar disks in this case.

\begin{fact}[\cite{abt-cbvdrt-08}]
    For any two points $\pntA,\pntB \in \Domain$, we have that
    $\distX{\pntA}{\pntB} \leq \distGeo{\pntA}{\pntB} \leq \slFactor
    \distX{\pntA}{\pntB}$.
    
    \factlab{geo:dist}
\end{fact}

\begin{lemma}
    Let $D$ be a geodesic disk of radius $r$ on the surface of a
    terrain with bounded slope $\slConst$ and let $A$ denote its area.
    We have that $\pi (r/\slFactor)^2 \leq A \leq \pi \slFactor r^2$.
    
    \lemlab{geo:disk:area}
\end{lemma}

\begin{proof}
    Let $c$ be the center of $D$ and let $c_p$ be its projection.  Let
    $D_{r/\slFactor}$ and $D_r$ be the planar disks with center $c_p$
    and radius $r/\slFactor$ and $r$, respectively.  Let
    $T_{r/\slFactor}$ and $T_r$ denote the portion of the terrain that
    lies directly above $D_{r/\slFactor}$ and $D_r$, respectively.
    
    Clearly the projection of $D$ is contained in $D_r$ and so if we
    can bound the area of $T_r$ then this will bound the area of $D$.
    Observe that $D_r$ consists of a set of triangles from
    $\Triangulation$, which have been clipped at the boundary of
    $D_r$.  It is a well-known fact that if we lift any such (clipped)
    triangle up to the terrain then its area can increase by at most a
    factor of $\slFactor$.  Therefore the total area of $T_r$ is at
    most $\pi \slFactor r^2$.
    
    Now consider the disk $D_{r/\slFactor}$.  First observe that
    $T_{r/\slFactor}$ must have area at least as large as that of
    $D_{r/\slFactor}$.  Second, note that $D$ must contain
    $T_{r/\slFactor}$ (since by \factref{geo:dist} the distance
    between any point $T_{r/\slFactor}$ and $c$ is at most
    $\slFactor(r/\slFactor) = r$).  Therefore the area of $D$ is at
    least $\pi(r/\slFactor)^2$.
\end{proof}

\subsection{Complexity of the Voronoi Diagram}
\seclab{def:complexity}

The complexity of the Voronoi diagram is measured by the complexity of
the structured set of curves and line segments that delineate the
Voronoi cells.  This set consists of pieces of bisectors and it can be
characterized as follows.  Again, we adopt the definitions used in
\cite{abt-cbvdrt-08}.

For most of the points on a bisector, the shortest path to either site
will be unique.  If the shortest path is not unique, we call $\pnt$ a
\emphi{breakpoint}.  The breakpoints partition the bisector into a set
of curved pieces which we call \emphi{\chords.} The
\emphi{combinatorial complexity} of the Voronoi diagram is now defined
as the sum of
\begin{inparaenum}[(i)]
    \item the number of Voronoi vertices,
    \item the number of breakpoints of Voronoi edges, and
    \item the number of intersections of the \chords of Voronoi edges
    with the triangulation of the terrain.
\end{inparaenum}

We continue with some useful facts and lemmas used in the analysis of
the complexity.  First, it was observed by Moet \etal that the number
of breakpoints of the Voronoi diagram is bounded by $n$, since each of
them can be attributed to a terrain vertex.  To see why this is true,
imagine walking along the bisector while sweeping the shortest paths
to either site from the current position. The points on the bisector,
where the edge sequence of the shortest path changes, are the
breakpoints.

\begin{fact}[\cite{mks-ort-08}]
    Given a terrain $\Terrain$ which is defined by a triangulation
    with $n$ vertices, the number of breakpoints of any Voronoi
    diagram on $\Terrain$ is smaller than or equal to $n$.

    \factlab{breakpoints}
\end{fact}

Furthermore, we will use the following result by Aronov \etal
\cite{abt-cbvdrt-08}.

\begin{lemma}[\cite{abt-cbvdrt-08}]
    Given two points $\pntA$ and $\pntB$, the set of \chords that form
    the bisector of $\pntA$ and $\pntB$ on $\Terrain$ (projected to
    the $(x,y)$-plane) is $O(\slConst)$-low density.

    \factlab{l:d:bisector}
\end{lemma}

We remark that Aronov \etal use this result to show that the bisector
has linear complexity.  However, note that this result does not imply
that the overall set of \chords of the Voronoi diagram is low density,
which would imply a linear complexity for the whole diagram.  Consider
for example the situation, where all the sites lie close to each other
on a straight line, and all the triangles of the terrain surface are
coplanar.  In this example, the bisectors are pairwise parallel lines,
which extend from one side of the domain to the other and could
therefore lead to a quadratic complexity Voronoi diagram by
intersecting many triangles of the terrain.

Finally, we observe that the number of Voronoi vertices and edges is
linear in the number of sites, as the following lemma and corollary
testify.  This fact was also observed by Aronov \etal, we include an
independent proof for the sake of completeness.

\begin{lemma}
    Let $\Terrain$ be a terrain and let $\Sites$ be a set of $m$
    points.  Then the number of Voronoi edges and Voronoi vertices of
    $\VD{\Sites}$ is $O(m)$.
    
    \lemlab{nr:V:e:v}
\end{lemma}
   
\begin{proof}
    For $m \leq 2$ the claim is clearly true, since we have at most one
    bisector, which contributes exactly one Voronoi edge.  For $m > 2$, we argue
    as follows.
    
    First, observe that the cells in this Voronoi diagram are
    connected. Indeed, consider a point $\pnt$ that belongs to the
    interior of the cell of $\site \in \Sites$. Consider the shortest
    path $\pi$ from $\pnt$ to $\site$, and consider any point $\pntA
    \in \pi$. If $\pntA$ is closer to some other site $\siteA$ than to
    $\site$, then we have that
    \begin{align*}
        \distGeo{\pnt}{\site} = %
        \distGeo{\pnt}{\pntA} + \distGeo{\pntA}{\site}%
        \geq \distGeo{\pnt}{\pntA} + \distGeo{\pntA}{\siteA}%
        \geq \distGeo{\pnt}{\siteA},
    \end{align*}
    but this is a contradiction to $\pnt$ being in the interior of the
    cell of $\site$.
    
    Now, consider the dual graph $\Graph$ of the graph formed by the
    Voronoi vertices and Voronoi edges. In this graph, every vertex
    corresponds to a Voronoi cell and every face corresponds to a
    Voronoi vertex.  Note that we can derive a geometric embedding of
    this graph by using the sites as vertices and picking an arbitrary
    point on each Voronoi edge and connecting it by its shortest path
    to either site to form a (possibly curved) edge between two
    vertices.
    
    It is well-known that a cell in the Voronoi diagram might not be
    simply connected. Indeed, consider a mountain surrounded by a
    plane. If we place a site $\site$ on the top of the mountain, and
    a site $\siteA$ at the bottom of the mountain (and the mountain
    slope is large enough) then the Voronoi cell of $\site$ would be
    completely surrounded by the cell of $\siteA$ and thus would
    create a 'hole' in this cell. In $\Graph$, the two vertices that
    correspond to the cells of $\site$ and $\siteA$ would be connected
    by an edge, which is incident to only one face in $\Graph$.
    
    Furthermore, it is known that the dual graph can have multiple
    edges between two sites. To see this, again, place $\site$ on the
    top of the mountain and place two sites $\siteA$ and $\siteB$ at
    the bottom, such that the bisector between $\siteA$ and $\siteB$
    intersects the mountain.  The boundary between the cells of
    $\siteA$ and $\siteB$ would contain two Voronoi edges from the
    same bisector.
    
    However, the dual graph $\Graph$ is planar and connected and as
    such its Euler characteristic is $2$. Hence $v-e+f=2$, where $e$
    denotes the number of edges, $f$ the number of faces and $v$ the
    number of vertices of $\Graph$.
        
    Now, by definition, every Voronoi vertex is incident to at least
    three Voronoi cells. Therefore, every face of $\Graph$ is incident
    to at least three edges of $\Graph$. Since $\Graph$ is planar,
    every edge of $\Graph$ is incident to at most two faces of
    $\Graph$.  Hence, we have that $3f \leq 2e$, which implies that
    $3f \leq 2(v+f-2)$, and therefore it holds that $f \leq 2v - 2 =
    2m-2$.  It follows that the number of Voronoi vertices is in
    $O(m)$.  Applying Euler's formula again, we obtain that also the
    number of Voronoi edges is in $O(m)$.
\end{proof}

\begin{corollary}
    Let $\Terrain$ be a terrain and let $\Sites$ be a set of $m$
    points.  Let $\Square$ be a sub-square which intersects $k$
    Voronoi cells in their projection.  The number of Voronoi edges
    which intersect $\VD{\Sites} \cap \Square$ in their projection is
    in $O(k)$.
    
    \corlab{nr:V:e:cell}
\end{corollary}

\section{Upper bound}
\seclab{upper:bound}

We prove the following lemma first in the planar case and then extend
it to terrains with bounded slope.  The bounded expected complexity
then follows by examining the number of intersections of the chords
with the terrain triangulation in an $\sqrt{m}\times\sqrt{m}$ grid.

\begin{lemma}
    \lemlab{square} Let $\Sample$ be a set of $m$ points, sampled
    uniformly at random from a unit square, and let $\Square$ be a
    sub-square contained in the unit square of side length
    $1/\sqrt{m}$.  Then the expected number of points in $\Sample$
    that contribute to $\VD{\Sample}\cap \Square$ is $O(1)$.
\end{lemma}

\begin{proof}
    We place a sequence of exponentially growing disks centered at the
    center point of $\Square$.  Let $r_i = \frac{1}{\sqrt{2m}}2^{i}$,
    for $i=0,\dots, k=\ceil{ \lg \sqrt{2m} \,}$ (i.e. $r_0$ is the
    radius of the circumscribed circle of $\Square$).  Let $\Disk_i$
    be the disk of radius $r_i$, which is clipped to the unit square
    and let $\ringX{i} = \Disk_i\setminus \Disk_{i-1}$, for
    $i=1,\dots, k$.
   
    \parpic[r]{ \includegraphics[width=0.2\linewidth]{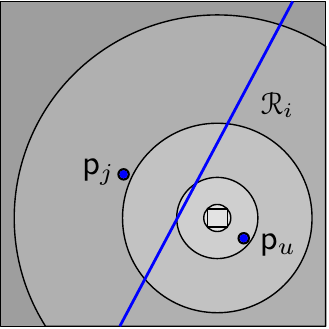}}
    
    Let the points in $\Sample$ be labeled $\pnt_1, \dots, \pnt_m$.
    Observe that the expected number of points from $\Sample$ that
    fall into $\Disk_2$ is $(\pi r_2^2)m = \frac{16\pi m}{2m} = 8\pi
    =O(1)$.  Hence we do not need to worry about their contribution to
    $\VD{\Sample}\cap \Square$. Otherwise, we claim that a point
    $\pnt_j$ which falls into $\ringX{i}$ for $i>2$, can only
    contribute to $\VD{\Sample}\cap \Square$ if $\Disk_{i-2}$ contains
    no points of $\Sample$. Assume for the sake of contradiction that
    the Voronoi cell of $\pnt_j$ intersects $\Square$, and there
    exists some point $\pnt_u$ in $\Sample$ that lies in
    $\Disk_{i-2}$.  By construction, we know $\pnt_j$ has distance
    greater than $(r_{i-1} - r_0)$ to any point in $\Disk_0$ and
    $\pnt_u$ has distance at most $(r_{i-2} + r_0)$ to any point in
    $\Disk_0$.  Hence we have that $\distSet{\pnt_j}{\Disk_0} >
    r_{i-1} - r_0 = 2 r_{i-2} - r_0 \geq r_{i-2} + r_0 \geq
    \distSet{\pnt_u}{ \Disk_1}$ for $i > 2$ and as such every point in
    $\Disk_0$ is strictly closer to $\pnt_u$ than $\pnt_j$, where
    $\distSet{\pnt}{X} = \min_{\pntA \in X} \distX{\pnt}{\pntA}$.
    
    Hence it is sufficient to bound the expected number of points
    $\pnt_j$ which fall into an annulus $\ringX{i}$ such that
    $\Disk_{i-2}$ is empty.  For $\pnt_j$ we define the indicator
    variable $X_i^j$ which is equal to $1$ if and only if $\pnt_j\in
    \ringX{i}$, and the indicator variable $Y_i^j$ which is equal to
    $1$ if and only if no other point falls into $\Disk_{i-2}$.  Hence
    a given point $\pnt_j$ can contribute to $\VD{\Sample}\cap
    \Square$ if and only if $X_i^jY_i^j=1$ for some value of $i$.
    Now, we know that
    \begin{align*}
        \Prob{X_i^j=1}%
        \leq%
        \pi r_i^2- \pi r_{i-1}^2%
        =%
        \frac{\pi}{2m}(2^{2i}-2^{2(i-1)})%
        =%
        \frac{3\pi}{2m}4^{i-1},
    \end{align*}
    and
    \begin{align*}
        \Prob{Y_i^j=1}%
        &\leq%
        \pth{1-\frac{1}{4}\pi r_{i-2}^2}^{m-1}%
        \leq%
        \exp\pth{-\frac{\pi}{4} r_{i-2}^2(m-1)}%
        =%
        \exp\pth{-\frac{\pi 4^{i-3}(m-1)}{2m}}%
        \leq%
        \exp(-4^{i-3}).
    \end{align*}
    where a factor of $1/4$ was added in the bound for $Y_i^j$ due to
    boundary effects that might arise from the position of $\Square$
    in the unit square.  Hence the number of points that can affect
    $\VD{\Sample}\cap \Square$ is bounded by $\sum_{j}\sum_{i>2}
    X_i^jY_i^j$, for which in expectation we have,
    \begin{align*}
        \Ex{\sum_{j} \sum_{i>2} X_i^jY_i^j}%
        &=%
        \sum_{j} \sum_{i>2} \Ex{X_i^j} \Ex{Y_i^j}%
        \leq%
        \sum_{j} \sum_{i>2} \frac{3\pi}{2m}4^{i-1} e^{-4^{i-3}}%
        =%
        \frac{3\pi}{2} \sum_{i>2} 4^{i-1} e^{-4^{i-3}} %
        =%
        O(1),
    \end{align*}
    by linearity of expectation and the independence of $X_i^j$ and
    $Y_i^j$ for all $i$ and $j$.
\end{proof}

\begin{lemma}
    Let $\Terrain$ be a terrain with bounded slope $\slConst$.  Let
    $\Sample$ be a random sample of $m$ points, either sampled
    uniformly from $\Terrain$, or uniformly from the unit square and
    then lifted vertically up to $\Terrain$.  Let $\Square$ be a
    sub-square contained in the unit square of side length
    $1/\sqrt{m}$.  Then the expected number of points in $\Sample$
    that contribute to the portion of the Voronoi diagram of $\Sample$
    on $\Terrain$ that lies above $\Square$ is
    $O\pth{\slTFactor}$.
    
    \lemlab{contribution:sites}
\end{lemma}

\begin{proof}
    Follows by a careful adaptation of the proof of \lemref{square}.
    To accommodate for the larger distances on the surface, we
    increase the radii of the disks slightly.
    
    So, let $r_i = \frac{1}{\sqrt{2m}}(2\slFactor + 1)^{i}$, for
    $i=1,\ldots,k$. Let $\Disk_i$ be the disk (in the plane) of radius
    $r_i$ centered at the center of $\Square$. And let $\ringX{i} =
    \Disk_i\setminus \Disk_{i-1}$, for $i=1,\dots, k$, as defined
    above.  For points $\pnt, \pntA \in \Disk_{i-2}$ and $\pntB \in
    \ringX{i} = \Disk_{i} \setminus \Disk_{i-1}$, we have that
    \begin{align*}
        \distGeo{\pnt}{\pntA}%
        \leq%
        \slFactor \distX{\pnt}{\pntA}%
        \leq%
        \slFactor 2 r_{i-2} %
        \leq %
        r_{i-1} - r_{i-2}%
        <%
        \distX{\pntA}{\pntB} %
        \leq%
        \distGeo{\pntA}{\pntB}.
    \end{align*}
    As such, as before, for a point in $\Sample \cap \ringX{i}$ to
    affect the Voronoi diagram on $\Square$ requires that
    $\Disk_{i-2}$ is empty of any points of $\Sample$.
    
    Now, consider the case that the points are sampled from the
    terrain.  Define $X_i^j$ and $Y_i^j$ as in \lemref{square}. Note
    that the area of the terrain can only increase from one by lifting
    the individual triangles.  We have that
    \begin{align*}
        \Prob{X_i^j=1}%
        &\leq%
        \frac{\text{area on terrain of } \ringX{i}}{\text{area of
              terrain}}%
        \leq%
        \slFactor \pth{ \pi r_i^2- \pi r_{i-1}^2 }%
        \leq%
        \frac{\pi \slFactor}{2m} \pth{ (2\slFactor +1)^{2i} -
           (2\slFactor +1)^{2i-2} }%
        \\%
        &\leq%
        \frac{\pi \slFactor}{2m} (2\slFactor +1)^{2i}.
    \end{align*}
    Similarly, since only a quarter of $\Disk_{i-2}$ might be in the
    terrain, the probability of this area is bounded from below by
    $\pi r_{i-2}^2/(4\slFactor)$. Plugging this into the analysis of
    \lemref{square}, we have
    \begin{align*}
        \Prob{Y_i^j=1}%
        &\leq%
        \pth{1-\frac{\pi r_{i-2}^2}{4\slFactor^2 } }^{m-1}%
        \leq%
        \exp \pth{-\frac{\pi r_{i-2}^2}{4\slFactor^2 }(m-1) }%
        \leq%
        \exp \pth{- (2\slFactor + 1)^{2i-5} }.
    \end{align*}
    As before, we thus have
    \begin{align*}
        \Ex{\sum_{j} \sum_{i>2} X_i^jY_i^j}%
        \leq%
        \sum_{j=1}^m \sum_{i>2} \frac{\pi \slFactor}{2m} (2\slFactor
        +1)^{2i} \cdot \exp \pth{- (2\slFactor + 1)^{2i-5} }%
        = O(1).
    \end{align*}
    The only missing component is bounding the expected number of
    points of $\Sample \cap \Disk_2$, as they can affect the Voronoi
    diagram in $\Square$. Arguing as above, this quantity is bounded
    by $m \slFactor \cdot (\text{area of }\Disk_2) \leq m \slFactor \pi
    r_2^2 \leq \frac{\slFactor}{2}(2\slFactor + 1)^{4} =
    O\pth{\slFactor^5}$.
    
    Note that if we drop the factors of $\slFactor$ in the bounds for
    $X_i^j$, $Y_i^j$, and the area of $\Disk_2$, then the above
    becomes a proof for the case when we sample from the unit square
\end{proof}

\begin{theorem}
    Let $\Terrain$ be a terrain. Let $\Sites$ be a random sample of
    $m$ points, either sampled uniformly from the surface of
    $\Terrain$, or uniformly from the domain and then lifted
    vertically up to the surface.  The expected combinatorial
    complexity of $\VD{\Sites}$ is in $O\pth{ \slConst
       \slTFactor\ldConst(n+m)}$.
    
    \thmlab{main}
\end{theorem}

\begin{proof}
    As described in \secref{def:complexity}, the combinatorial
    complexity of the Voronoi diagram is the sum of the number of
    breakpoints of Voronoi edges, the number of Voronoi vertices and
    the number of intersections of triangulation edges with \chords of
    Voronoi edges. By \factref{breakpoints} the number of breakpoints
    is bounded by $O(n)$, and by \lemref{nr:V:e:v} the number of
    Voronoi vertices is bounded by $O(m)$.
    
    It remains to bound the number of intersections of the set of
    \chords with the triangulation.  To this end, we place a grid on
    the domain of the terrain, such that the side length of each grid
    cell is $l=1/\sqrt{m}$ and we obtain $M=O(m)$ grid cells which
    together cover the domain of the terrain.  Now, let
    $\Cell_1,\dots,\Cell_M$ denote these grid cells. Consider the grid
    cell $\Cell_i$ and the set of \chords of the Voronoi diagram which
    intersect this grid cell, let this set be $\Chords_i$. Similarly,
    let $\Edges_i$ denote the subset of edges of the triangulation,
    which intersect $\Cell_i$. Since we assumed that the triangulation
    is $\ldConst$-low density, also $\Edges_i$ is a $\ldConst$-low
    density set. By \factref{l:d:bisector} we have that the set of
    \chords, which originate from the same Voronoi edge (and therefore
    from the same bisector) form an $O(\slConst)$-low density set.  Let
    $k_i$ denote the number of Voronoi edges that contribute \chords
    to $\Chords_i$, we have that $\Chords_i$ is a $O(\slConst k_i)$-low
    density set.  By \factref{l:d:intersections}, the number of
    intersections between objects of $\Edges_i$ and objects of
    $\Chords_i$ is in $O(\slConst k_i\cardin{\Edges_i} + \ldConst
    \cardin{\Chords_i})$.
    
    Now, in order to bound the overall number of intersections, let
    $\Chords_i^{>l}$ denote the subset of \chords which are longer
    than $l$, similarly, let $\Chords_i^{\leq l}$ denote the \chords
    in $\Chords_i$ which have length smaller or equal to $l$ and let
    $\Edges_i^{\leq l}$ and $\Edges_i^{> l}$ be defined analogously.
    By the above analysis, we have that there exists some constant
    $\constA$, such that it holds for the overall number of
    intersections $\Intersections$,
    \begin{align*}
        \Intersections \leq \constA \sum_{i \geq 1}^{M}
        \pth{\slConst k_i \cardin{\Edges_i} + \ldConst
           \cardin{\Chords_i}} = \constA \sum_{i \geq 1}^{M}
        \pth{\slConst k_i\pth{\cardin{\Edges_i^{> l}}
              +\cardin{\Edges_i^{\leq l}}} +
           \ldConst\pth{\cardin{\Chords_i^{> l}} +
              \cardin{\Chords_i^{\leq l}}}}.
    \end{align*}
    By the definition of low density sets, we have that
    $\cardin{\Edges_i^{> l}} = O(\ldConst)$ and $\cardin{\Chords_i^{>
          l}} = O(\slConst k_i)$, since they intersect the bounding
    ball of the grid cell $\Cell_i$, which has radius $O(l)$.
    Therefore, it must be that there exists a constant $\constB$ such
    that,
    \begin{align*}
        \Intersections%
        \leq%
        \constA \pth{\sum_{i \geq 1}^{M} \slConst
           k_i\cardin{\Edges_i^{\leq l}} +
           \ldConst\cardin{\Chords_i^{\leq l}} + c_2\ldConst\slConst
           k_i}%
        \leq%
        \constA \pth{\sum_{i \geq 1}^{M} \slConst
           k_i\cardin{\Edges_i^{\leq l}} + c_2\ldConst\slConst
           k_i}+\constA \ldConst 4\cardin{\Chords},
    \end{align*}
    where the last inequality follows from the fact that any chord in
    $\Chords_i^{\leq l}$ can intersect at most four grid cells, since
    the grid cells have side length equal to $l$ (similarly any edge
    in $\Edges_i^{\leq l}$ can intersect at most three grid cells).
    
    Finally, note that the number of Voronoi cells that are expected
    to intersect a grid cell in their projection is bounded by
    $O(\slTFactor)$ by \lemref{contribution:sites}. Thus by
    \corref{nr:V:e:cell} we have that $\Ex{k_i} = O(\slTFactor) $.
    Therefore in expectation,
    \begin{align*}
        \Ex{\Intersections} &\leq%
        \constA \Ex{\sum_{i \geq 1}^{M} \pth{ \slConst k_i\cardin{
                 \Edges_i^{\leq l}} + c_2\ldConst\slConst k_i}} +
        \constA \ldConst 4\cardin{\Chords}%
        =%
        \constA \sum_{i \geq 1}^{M} \pth{ \slConst
           \Ex{k_i}\cardin{\Edges_i^{\leq l}} + c_2\ldConst\slConst
           \Ex{k_i}}
        + \constA \ldConst 4\cardin{\Chords} \\
        &\leq%
        \constC \slConst \slTFactor \ldConst \pth{\sum_{i \geq 1}^{M}
           \pth{\cardin{\Edges_i^{\leq l}} + 1}+ 4\cardin{\Chords}}
        \leq \constC \slConst \slTFactor \ldConst
        (3\cardin{\Edges}+M+4\cardin{\Chords})
    \end{align*}
    for some constant $\constC$ (note that we used the fact that
    $\cardin{\Edges_i^{\leq l}}$ is independent of the random
    sampling).  Furthermore, observe that by \lemref{nr:V:e:v} the
    overall number of Voronoi edges is $O(m)$. Recall that every
    Voronoi edge is broken up by breakpoints into \chords.  Every
    breakpoint increases the number of \chords by one. Using
    \factref{breakpoints}, it follows that the overall number of
    \chords $\cardin{\Chords}$ is in $O(n+m)$. Since the number of
    edges of the triangulation $\cardin{\Edges}$ is $O(n)$, we
    conclude that $\Ex{\Intersections} = O(\slConst \slTFactor
    \ldConst(n+m))$.
\end{proof}

\section{Lower bound}
\seclab{lower:bound}

In this section we show that if we drop the assumptions on the
terrain, then the expected worst-case complexity of the resulting
geodesic Voronoi diagram can be $\Omega\pth{n m^{2/3}}$ if the sites
are sampled uniformly at random from the unit square.

In the following we will refer to the walls of the unit square defined
by $x=0$, $x=1$, $y=0$, and $y=1$ as the \emphi{\west, \east,
\south, and \north walls}, respectively.


\begin{example}
    We start with the easy case of planar map and points in the unit
    square.

    \parpic[r]{%
           \includegraphics[width=.2\linewidth]{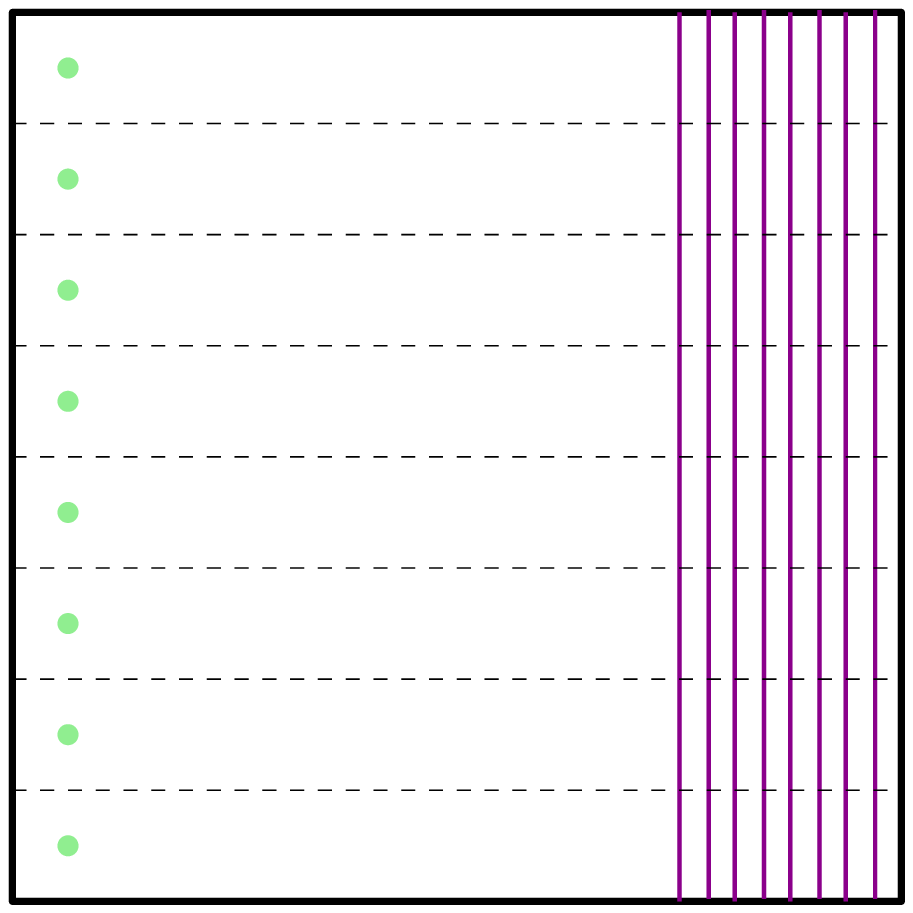} }

    First place $m$ points in a column near and parallel to the \west
    wall of the unit square such that the spacing between each
    adjacent pair of points is $\Theta\pth{1/m}$ (we assume $m=O(n)$).
    The planar map consists of $n$ vertical lines near the \east wall
    of the unit square that extend from the \north wall to the \south wall.
    Now, the boundaries of the Voronoi cells of these points extend from the
    west to the east wall and are parallel to the north and south walls, and
    hence the complexity of the overlay of the Voronoi diagram with the planar
    map is $\Theta\pth{nm}$, since it is an $n \times m$ grid.

    The later constructions use this as their starting point.

    \exelab{planar}
\end{example}

\subsection{Farming -- an $\Omega\pth{n \sqrt{m}\ts \ts }$ %
   example}
\seclab{farming}

\subsubsection{Construction}

The height function used in the following construction of a terrain
has (essentially) only two values, zero and $h$.  The areas between a
part of the terrain that is of height zero and of height $h$ consist
of very narrow and steep boundary regions. This intermediate boundary
would have a very small measure in the projection, and the reader can
think of it as having measure zero.  Moreover, $h$ is chosen to be
sufficiently large so that no point at height zero can affect the
Voronoi diagram at height $h$.  One can therefore view the following
terrain construction as a flat unit square, where we have cut out or
``forbidden'' areas (that have height $0$).  Therefore, for the sake
of simplicity of exposition, an area being constructed is flat, at
height $h$, and the adjacent forbidden area is at height zero.  Our
main building blocks will be farms.  We define a \emphi{farm} to be a
square of side length $1/(c\sqrt{m})$.  Intuitively, farms are part of
the terrain which with constant probability (the constant will depend
on $c$) will receive at least one point from the random sample (i.e. a
farm takes the place of a point from the example in \exeref{planar}).
We define the \emphi{diameter} of a farm to be the quantity
$\diamConst = \sqrt{2}/(c\sqrt{m})$ (that is, the distance of the
furthest two points in a farm).

\parpic[r]{%
   \begin{minipage}{0.2\linewidth}%
       \centerline{\includegraphics[width=.95\linewidth]{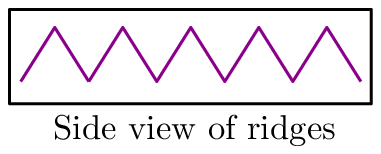}}%
   \end{minipage}}

We now define a sequence of ridges to take the place of the planar map
from \exeref{planar} (i.e. in expectation we would like the Voronoi
diagram to look like a grid over the ridges).  Formally, let a
\emphi{sequence of ridges} of length $n$ be a sequence of $2n$
rectangles, $r_1,\dots,r_{2m}$ such that the right edge of $r_i$ is
the same as the \west edge of $r_{i+1}$ for $i=1,\dots, 2n-1$, $r_i$
has a slope of $45^\circ$ for odd $i$ and $-45^\circ$ for even $i$,
all rectangles extend from the \north to \south walls of the unit
square, and the geodesic distance from the left edge of $r_1$ to the
right edge of $r_{2n}$ is $1/(c 2^n)$ (which is $O(1/2^m)$ since we
assumed $m=O(n)$).  Refer to the figure above to the right for a side
view of the ridges.

\parpic[r]{
   \includegraphics[width=.2\linewidth]{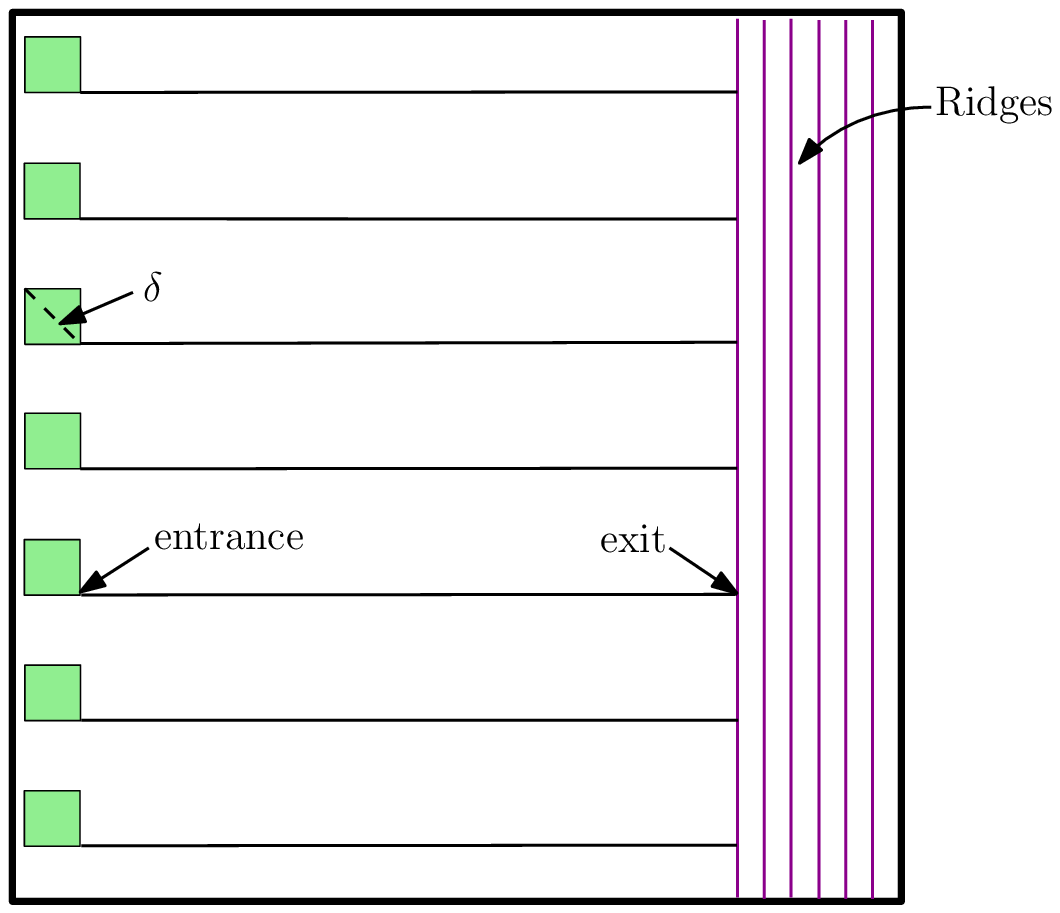}}

The construction of the $\Omega\pth{n \;\mysqrt{m} \ts}$ example is as
follows.  Place $\Theta\pth{\sqrt{m} \ts}$ farms from \north to \south
along the \west wall of the unit square, with $2/(c \sqrt{m})$ spacing
in between each adjacent pair.  Next build a sequence of
$\Theta\pth{n}$ ridges near (and parallel to) the \east wall of the
unit square.  Then connect each farm directly to the leftmost ridge by
creating a line parallel to the north and south walls connecting the
\southeast corner of the farm to the first ridge. See figure on the
right.  We refer to such a line as \emphi{road}. The roads stay at
height $h$ and to the left and right of a road, the height drops to
zero as described earlier.


\subsubsection{Analysis}

\begin{definition}
    The point at which a farm connects to its road is its
    \emphi{entrance} (i.e. the \southeast corner of the farm), and
    the point at which the road connects to the leftmost ridge is its
    \emphi{exit}.  We say that the point (from the random sample of
    $m$ points) that is closest to the entrance for some farm, is that
    farm's \emphi{dominating point}.
    Let $p$ and $e$ be the dominating point and exit, respectively, of
    some farm.  We say that another point $q$ from the random sample
    that is contained in another farm and such that $\distGeo{q}{e} <
    \distGeo{p}{e}$, \emphi{eliminates} (the Voronoi cell of) $p$,
    where $\distGeo{p}{e}$ denotes the shortest path on the terrain
    from $p$ to $e$.  If there are no points which eliminate a given
    dominating point, then the dominating point is \emphi{alive}.

    \deflab{eliminates}
\end{definition}

\begin{lemma}\lemlab{lower:bound:simple}
    For the construction of the terrain described above, if one picks
    uniform at random $m$ points in the unit square, their induced
    geodesic Voronoi diagram on this terrain has complexity
    $\Omega\pth{ n \sqrt{m} \ts }$.
\end{lemma}

\begin{proof}
    The area of each farm is $\Theta\pth{1/m}$ and hence a sample of
    $m$ points picked uniformly at random from the unit square, will
    have at least one point with constant probability in each farm.
    Moreover, since we constructed $\Theta\pth{\sqrt{m}}$ farms, this
    implies that in expectation $\Theta\pth{\sqrt{m}}$ farms will
    receive at least one point.  Furthermore, the width of the sequence of
    ridges and roads was chosen such that the probability that either
    receives a point is exponentially small (and hence in the
    following we assume they do not receive any point).  
    
    Now consider a farm which received at least one point, and let $p$ be its
    dominating point.
    Observe that the Voronoi cell of $p$ contains the entire road
    connecting this farm to the ridges, and its Voronoi cell extends
    all the way to the rightmost edge of the sequence of ridges, and hence will
    be of complexity $\Omega\pth{n}$.  Indeed, by our construction, only
    a point from another farm can prevent the Voronoi cell of $p$ from
    reaching the rightmost ridge.  However, the spacing of the farms
    was chosen to prevent this.  In the worst case $p$ is in the \northwest
    corner of its farm, and an adjacent farm has a point $q$ at
    the \southeast corner.  Let $l$ be the length of a road.  Let
    $e_p$ (resp. $e_q$) be the exit of the farm containing $p$ (resp
    $q$).  Now consider the geodesic shortest path connecting
    $e_p$ to the rightmost ridge.  Every point on this segment is in distance at
    most $\diamConst+l+1/(c 2^n)$ from $p$.  However, the
    closest point on this segment from $q$ is at a distance of at
    least $l+2/(c\sqrt{m}) \geq \diamConst+l+1/(c 2^n)$,
    and so $q$ cannot prevent the Voronoi cell of $p$ from reaching
    all the way to the rightmost ridge.

    Therefore, in expectation, we have that $\Theta\pth{\sqrt{m}}$
    farms have a point whose Voronoi cell extends all the way across the
    sequence of ridge, which gives a Voronoi diagram that in expectation
    has complexity $\Omega\pth{n\sqrt{m}}$.
\end{proof}

\subsection{Industrial farming -- an %
   $\Omega\pth{n m^{2/3}\ts \ts }$ %
   example}

The challenge in improving the example above is that the distance of a
dominating point to the exit of a farm has too much variance (i.e.,
$\sqrt{1/m}$). Since there does not seem to be a way to decrease the variance
directly, instead we connect all the farms to the ridges, and carefully argue
about the expected complexity of the generated Voronoi diagram.

\subsubsection{Construction}

In the following we assume that $m = O(n)$.  

We set the side length of each square farm to be $1/ \sqrt{m}$ and
construct a sequence of $\Theta\pth{n}$ ridges near the \east wall of
the unit square.  We will place an $M \times M$ grid of farms inside
the unit square, where $M = \floor{\sqrt{m}/4}$.  Specifically, the
spacing between columns (which extend from north to south) will be 
$1/\sqrt{m}$ and the spacing between rows (which extend from west to 
east) will be $2/\sqrt{m}$.
The grid starts in the \northwest corner of the unit square.

We now describe the connecting roads from the farms to the ridges.
The following construction of the roads will ensure that the length of
each road is the same and that the distance between adjacent exits on
the ridges is at least $1/m$.  These two properties will be sufficient
for the analysis in the next section to go through.

\parpic[r]{%
   \includegraphics[width=.3\linewidth]{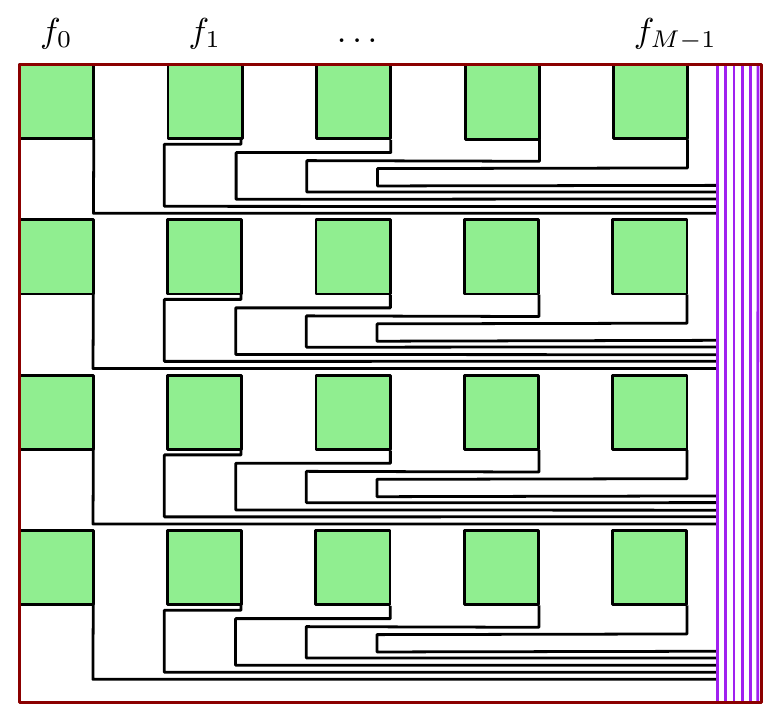} }

Consider a given row of farms.  Number the farms in this row $f_0,
\dots, f_{M-1}$ in increasing order of their distance to the \west
wall. Every farm has dimensions $1/\sqrt{m} \times 1/\sqrt{m}$, and the spacing
between two consecutive farms in a row is $1/\sqrt{m}$.  As such,
the $x$ coordinate of the entrance of the $i$\th farm is $x_i =
(2i+1)/\sqrt{m}$ (as before, the entrance to each road will be at the
\southeast corner of the farm).  The directions the $i$\th farm's
road goes from entrance to exit is described as follows:
\begin{compactenum}[\qquad(a)]
    \item \south for a distance of $\alpha_{i,a} = i/m$,
    \item \west for a distance of $\alpha_{i,b} = (x_i + \alpha_{i,a})/2$,
    \item \south for a distance of $\alpha_{i,c} = 1/\sqrt{m}-2\alpha_{i,a}$, and
    \item \east for a distance of $\alpha_{i,d}= w - (x_i-\alpha_{i,b})$ 
    all the way to the first ridge, where $w$ is the distance from the
    \west wall to the first ridge.
\end{compactenum}

This layout is sketched in the figure above to the right. Note that
the spacing in this figure only approximately matches the description.

\paragraph{Sanity checks.}
The road of the $i$\th farm starts at $x$ coordinate $x_i$, goes \west
for a distance of $\alpha_{i,b}$ and \east for a distance of
$\alpha_{i,d}$. Observe that the $x$ coordinate of the exit of this
road is $x_i - \alpha_{i,b} +\alpha_{i,d} = x_i -\alpha_{i,b} + w -
(x_i-\alpha_{i,b}) = w$.

Observe that the $i$\th farm will connect to the ridges in north to south
distance $\alpha_{i,a} + \alpha_{i,c} = 1/\sqrt{m}-i/m$ from the
southern boundary of the row of farms.  That is, adjacent farms in the row have
exits in distance $i/m$ apart along the first ridge (exits are
$\Theta\pth{1/\sqrt{m}}$ apart between rows).  Furthermore, each
road is of the same length. Indeed, let $r_i$ be the length of the
road for the $i$\th farm to its exit. We have that
\begin{align*}
    r_i &=
    \alpha_{i,a}%
    + \underbrace{\frac{{x_i + \alpha_{i,a}}}{2}}_{\alpha_{i,b}}%
    + \underbrace{\frac{1}{\sqrt{m}} - 2\alpha_{i,a}}_{\alpha_{i,c}} %
    + \underbrace{w - \pth{x_i-\frac{{x_i +
                \alpha_{i,a}}}{2}}}_{\alpha_{i,d}}%
    =%
    \frac{1}{\sqrt{m}} + w.
\end{align*}
That is, all the roads have exactly the same length.

\subsubsection{Competing farms}
\seclab{ideal}

We now prove that in expectation $\Theta\pth{m^{2/3}}$ dominating
points will have their Voronoi cells reach all the way to the \east
wall across the sequence of ridges.

\begin{observation}
    Let $p$ and $e_p$ be the dominating point and exit, respectively,
    of some farm $f$.  Let $q$ be a point which is in some other farm
    $f'$ with exit $e_q$.  If $f'$ is $i$ farms away (in the north 
    to south order of the exits of the farms along the first ridge) then
    $\distX{e_p}{e_q} = i/m$ and hence $q$ eliminates $p$ if and only if
    $\distGeo{q}{e_p} = \distGeo{q}{e_q} + i/m < \distGeo{p}{e_p}$.
    
    \obslab{ifarms}
\end{observation}

Next, we prove that each dominating point is alive with probability
$\Omega\pth{1/m^{1/3}}$, using the following helper lemmas.

\begin{lemma}
    Let $X$ be a positive random variable with expected value $\mu$.
    We then have that $\Ex{e^{-X}} \geq e^{-2\mu}/2$
    
    \lemlab{exp}
\end{lemma}

\begin{proof}
    Markov's inequality implies that $\Prob{X< 2\mu } = 1-\Prob{X\geq
       2\mu } \geq 1-\mu/2\mu=1/2$.  Therefore, by the definition of
    expectation, we get $\Ex{e^{-X}} \geq%
    \Prob{X\geq 2\mu} * 0 + \Prob{ X < 2\mu } *e^{-2\mu} \geq
    e^{-2\mu}/2$.
\end{proof}

\begin{lemma}
    Let $f$ be a farm, and let $r(f)$ be the random variable that is
    the distance of the closest site (that falls into this farm) from
    the farm entrance (if there is no site in this farm, we set $r(f)
    =\sqrt{2/m}$).  Then, for any distance $s$, we have
    $\Prob{\MakeSBig r(f)\leq s} \leq ms^2\pi/4$, where equality holds
    for $0\leq s\leq 1/\sqrt{m}$.

    \lemlab{distance}
\end{lemma}

\begin{proof}
    Recall that the entrance of a farm is at the \southeast corner.
    Hence a point in the farm which is in distance at most $s$
    from the entrance must fall into the intersection of the farm with
    a circle of radius $s$ whose center is at the entrance of the
    farm, see figure.

    \parpic[r]{%
       \begin{minipage}{0.18\linewidth}%
           \includegraphics{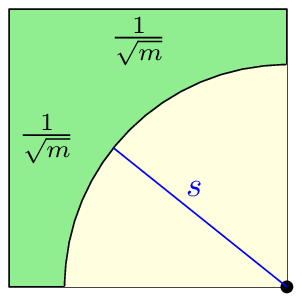}%
       \end{minipage}}

    Therefore, if the radius of the circle is less than the side
    length of the square, i.e. $s\leq 1/\sqrt{m}$, then the
    intersection is a quarter disk and so the area is exactly $\pi
    s^2/4$.  Otherwise, the top and left portions of the quarter disk
    will be needed to be clipped to the farm and so the area is $\leq
    \pi s^2/4$. Now, the probability of the $i$\th site to fall into
    this disk is $\leq \pi s^2/4$, and since we sample $m$ sites
    (independently), by the union bound the claim follows.
\end{proof}

\begin{lemma}
    Let $p$ and $e$ be the dominating point and exit, respectively, of
    a farm $f$.  Let $r=\distGeo{p}{e}$.  Let $f_i$ be a farm which is
    $i$ farms away from $f$ (either in north or in south direction), and let
    $X_i$ be the number of points which fell into $f_i$.  Let $\alpha_{X_i}$
    denote the probability that no point from $f_i$ eliminates $p$
    (see \defref{eliminates}). Then $\alpha_{X_i} \geq
    \exp\pth{-m(r-i/m)^2 X_i/2}$.

    \lemlab{alpha}
\end{lemma}

\begin{proof}
    Let $q_1,\dots, q_{X_i}$ be the $X_i$ points that fall into farm
    $f_i$.  Let $e_i$ be the exit of $f_i$, and let $d_j =
    \distGeo{q_j}{e_i}$, for $j=1,\dots,X_i$.  Arguing as in
    \lemref{distance}, we have that $\Prob{d_j\leq s} \leq s^2\pi/4$
    for all $j$.  By \obsref{ifarms}, a point $q_j$ eliminates $p$ if
    and only if $d_j< r-i/m$.  For $j \neq l$, whether or not $q_j$ or
    $q_l$ kill $p$ are independent events and hence
    \begin{align*}
        \alpha_{X_i} &=%
        \prod_{j=1}^{X_i} \Prob{ \MakeBig d_j\geq r-i/m}%
        \geq%
        \pth[]{1 - \frac{(r-i/m)^2 \frac{\pi}{4}} {\mathrm{area}(f)}
        }^{X_i}%
        \geq%
        \exp\pth{-m(r-i/m)^2\frac{\pi X_i}{2}}.
    \end{align*}
    \aftermathA
\end{proof}

\begin{lemma}
    Let $p$ and $e$ be the dominating point and exit, respectively, of
    some farm $f$.  Let $r=\distGeo{p}{e}$.  Let $X_i$ (resp. $Y_i$),
    for $i=1,\dots, \floor{ r m }$, denote the number of points which
    fall into the farm which is $i$ farms to the north (resp. south),
    from $f$ in the order of the exits along the first ridge.  Then
    the probability that $p$ is not eliminated given the values of
    $X_i$ and $Y_i$ for all $i$, is at least
    \begin{align*}
        \exp\pth{-m\sum_{i=1}^{\floor{ rm }} (r-i/m)^2\pi (X_i+Y_i)/2}
    \end{align*}%
    
    \lemlab{given}
\end{lemma}

\begin{proof}
    First note that for $i'> \floor{ rm }$, we have that $i'/m \geq
    (\floor{ rm } + 1) / m > r$.  Namely no point from a farm $i'$
    farms away can kill $p$, and hence we can ignore such farms.
    
    Given the value $X_i$ and $Y_i$ for all $i$, whether a farm
    contains a point which eliminates $p$ is independent from whether
    any other farm contains a point which eliminates $p$. Therefore,
    by \lemref{alpha},
    \begin{align*}
        &\hspace{-1cm} \Prob{\MakeBig \text{$p$ is alive} \sep{ X_1,
              \dots, X_{\floor{ rm }}, Y_1, \dots, Y_{\floor{ rm }}}}%
        =%
        \pth{\prod_{i=1}^{\floor{ rm }} \alpha_{X_i}}
        \pth{\prod_{i=1}^{\floor{ rm }} \alpha_{Y_i}} \\
        &\geq \pth{\prod_{i=1}^{\floor{ rm }} \exp\pth{-m(r-i/m)^2\pi
              X_i/2} } \pth{\prod_{i=1}^{\floor{ rm }}
           \exp\pth{-m(r-i/m)^2\pi Y_i/2}}%
        \\%
        &=%
        \exp\pth{-m\sum_{i=1}^{\floor{ rm }} (r-i/m)^2 \pi(X_i+Y_i)/2}
    \end{align*}
    \aftermathA
\end{proof}

\begin{lemma}
    Let $p$ and $e$ be the dominating point and exit, respectively, of
    some farm $f$, and let $r=\distGeo{p}{e}$.  Then the probability
    that $p$ is alive is $\geq \frac{1}{2}\exp\pth{-2r^3m^2}$.
    
    \lemlab{alive}
\end{lemma}

\begin{proof}
    Let $X_i$ and $Y_i$, for $i=1,\dots,\floor{ rm }$, be random
    variables equal to the number of points which fall into the farm
    which is $i$ farms to the north or south, respectively, from $f$
    in the order of the exits along the first ridge
    (note that if there is no farm $i$ farms to the north or south,
    then $X_i=0$ or $Y_i=0$, respectively).
    
    \lemref{given} tells us that $\Prob{p \text{ is alive} \sep{ X_1,
          \dots, X_{\floor{ rm }}, Y_1, \dots, Y_{\floor{ rm }}}} \geq
    e^{-T}$, where $T=\sum_{i=1}^{\floor{ rm }} (r-i/m)^2
    \pi(X_i+Y_i)/2$.  Since the area of each farm is $1/m$, We know
    that $\Ex{X_i} \leq 1$ and $\Ex{Y_i } \leq 1$ for
    $i=1,\dots,\floor{ rm }$ (``$\leq 1$" is used instead of ``$=1$"
    since there might not be a farm at that distance).  Therefore,
    \begin{align*}
        \Ex{T}%
        &=%
        \Ex{\sum_{i=1}^{\floor{ rm }} m(r-i/m)^2 \pi(X_i+Y_i)/2} =%
        \sum_{i=1}^{\floor{ rm }} m(r-i/m)^2 \pi (\Ex{X_i} +
        \Ex{Y_i})/2\\%
        &\leq%
        \sum_{i=1}^{\floor{ rm }} m(r-i/m)^2 \leq m(r^2m) = r^3m^2
    \end{align*}
    
    By \lemref{exp},
    $\ds \Prob{ \MakeBig p \text{ is alive}}%
    \geq%
    \Ex{e^{-T}} \geq \frac{1}{2}\exp\pth{-2\Ex{T}}%
    \geq%
    \frac{1}{2}\exp\pth{-2r^3m^2}.  $
    
    \aftermathA
\end{proof}

\begin{observation}
    \lemref{alive} implies that if $r\leq 1/m^{2/3}$ then %
    \begin{align*}
        \Prob{p \text{ is alive}}%
        \geq%
        \frac{1}{2}\exp\pth{-2r^3m^2} \geq \frac{1}{2e^2}.
    \end{align*}
    
    \obslab{smallr}
\end{observation}

\begin{lemma}
    The probability that a farm $f$ gives rise to a Voronoi cell that
    is not eliminated is $\geq \pi/(8e^2m^{1/3})$.
    
    \lemlab{alive2}
\end{lemma}

\begin{proof}
    By the equality in \lemref{distance} for short distances, we know
    that $\Prob{ r(f) \leq s } = ms^2\pi/4$, for $0\leq s\leq
    1/\sqrt{m}$.  Let $p$ be the dominating point of $f$. By
    \obsref{smallr} and \lemref{distance}, we have
    \begin{align*}
        \Prob{ p \text{ is alive}}%
        &\geq%
        \Prob{\MakeBig \pth{r(f) \leq m^{-2/3}} \cap \pth{
              p\text{ is alive}}}%
        =%
        \Prob{\MakeBig { \; p\text{ is alive}}
           \sep{ r(f) \leq m^{-2/3}}} \Prob{  r(f) \leq m^{-2/3} }\\
        &\geq%
        \frac{1}{2e^2} \Prob{\MakeBig r(f) \leq m^{-2/3}} %
        \leq%
        \frac{1}{2e^2} \cdot \frac{\pi m}{4m^{4/3}} =
        \frac{\pi}{8e^2m^{1/3}}
    \end{align*}
    \aftermathA
\end{proof}

\begin{theorem}
    In expectation, the Voronoi diagram will be of complexity
    $\Omega\pth{nm^{2/3}}$.
    
    \thmlab{lower:bound}
\end{theorem}

\begin{proof}
    Every farm which receives a point from the random sample has a
    dominating point.  Since $\Theta\pth{m}$ farms were built, and
    each farm receives one point in expectation, the expected number
    of dominating points is $\Theta\pth{m}$.  Therefore, by
    \lemref{alive2}, the expected number of alive dominating points is
    $\Omega\pth{m*(\pi/(8e^{2}m^{1/3}))} = \Omega\pth{m^{2/3}}$.
    
    Given that a dominating point is alive, the probability that its
    Voronoi cell does not reach the rightmost ridge is negligible.
    Therefore, in expectation, if there are $\Omega\pth{m^{2/3}}$
    alive dominating points and $\Theta\pth{n}$ ridges then the
    complexity of the Voronoi diagram will be $\Omega\pth{nm^{2/3}}$.
\end{proof}

\section{Conclusions}
\seclab{conclusions}

We investigated the expected combinatorial complexity of geodesic
Voronoi diagrams on polyhedral terrains in two settings where the
sites are being picked randomly. Usually, such random settings are the
great simplifier -- for example, the expected complexity of the convex
hull of $n$ points picked uniformly in the unit square is $O( \log n)$
-- but in our case the situation is considerably more subtle.

We proved that the expected complexity is linear if one assumes low
density and bounded slope and the domain of the terrain is a unit
square.  On the other hand, we described a worst-case construction of a
terrain which implies a super-linear lower bound on the expected
complexity if these assumptions are dropped. This implies that the
probabilistic analysis alone does not yield a linear complexity.

\parpic[r]{\includegraphics[width=0.2\linewidth]{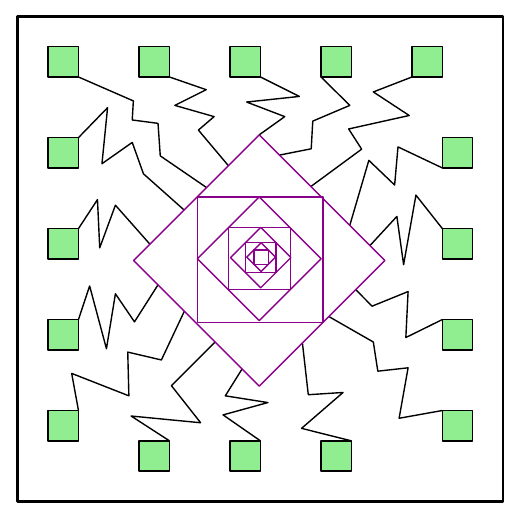}}

There are still many interesting open question for further research.
In particular, if we relax the realistic input assumptions, is the
expected complexity still linear or can one show other lower bounds?
One could, for instance allow the terrain to have a constant number of
triangles, where the slope is unbounded, or drop the steepness
assumption completely.  Consider the farming layout to the right,
where the sequence of ridges have been replaced by a recursive low
density construction. One problem with this construction is that the
low density assumption requires the ridges to have a non-negligible
area, which could catch points from the random sample.

\bibliographystyle{alpha}%
\bibliography{shortcuts,geometry}

\appendix

\section{How fat are Random Voronoi Cells?}
\apndlab{fat:cells}

In light of the known results on Poisson Voronoi Diagrams
\cite{obsc-stcav-00}, the results in this appendix are not too
surprising. We include our analysis here because it is relatively
simple and it is self contained. We are unaware of any work analyzing
the fatness of Voronoi cells. Furthermore, our sampling model is
different than the one used in the Poisson Voronoi Diagrams.

Let $\PntSet$ be a set of $m$ points picked uniformly from the unit
square $[0,1]^2$. To simplify the presentation, we assume the square
has the torus topology (i.e., identifying together facing edges). Let
$\pnt_1$ be the first sampled point, and let $\PntSet = \brc{\pnt_1,
   \ldots, \pnt_m}$. In the following, let $\VorCell{\pnt_1}{\PntSet}$
denote the Voronoi cell of $\pnt_1$ in the Voronoi diagram of
$\PntSet$.

\begin{lemma}
    For any integer $j > 0$, with probability $\geq 1-6 e^{1-j}$, the
    diameter of $\Cell = \VorCell{\pnt_1}{\PntSet}$ is bounded by $R_j
    = 4\cdot 3^{-1/4} \, \mysqrt{j/(m-1)}$.
    
    \lemlab{upper:bound}
\end{lemma}

\begin{proof}
    Partition the plane around $\pnt_1$ into $6$ equally spaced cones,
    of 60 degrees each. Consider such a cone, and break it into tiles
    of area $1/(m-1)$ each, by cutting it by parallel lines forming
    equal angles with both sides of the cone. In particular, let $T_i$
    be the $i$\th tile in this partition, with the tiles ordered in
    increasing order way from the center $\pnt_1$.
    
    \parpic[r]{\includegraphics{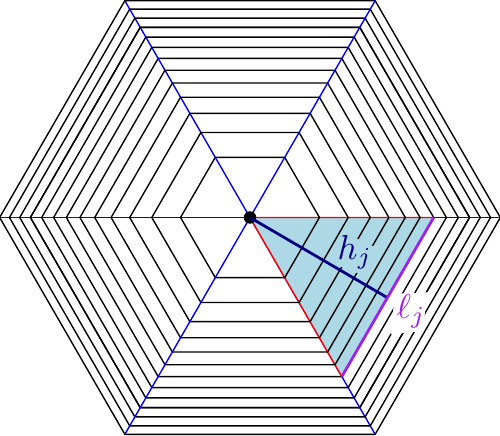}}
    
    The first $j$ tiles of this cone form an equilateral triangle of
    area $j/(m-1)$. As such, letting $\ell_j$ and $h_j$ be the edge
    length and height of this triangle, we have that $h_j =
    (\sqrt{3}/{2})\ell_j$, and its area is $h_j \ell_j /2 = \ell_j^2
    \sqrt{3}/4 = j/(m-1)$, which implies that $\ell_j = 2\cdot
    3^{-1/4} \mysqrt{j/(m-1)}$.
    
    Clearly, the probability that the first $j$ tiles would not
    contain any point of $\PntSet$ is $(1-j/(m-1))^{m-1} \leq
    \exp(-j(m-1)/(m-1)) \leq \exp(1-j)$.  As such, the probability
    that the first $j$ tiles in any of these six cones is empty is
    bounded by $6 \exp(1-j)$. Namely, with probability $\geq 1-6
    \exp(1-j)$, we have that all cones have a point in them in
    distance at most $\ell_j$ from $\pnt_1$. This implies that the
    Voronoi cell of $\pnt_j$ in such a cone in contained the in first
    $j$ tiles. As such, the diameter of this cell is bounded by
    $2\ell_j$, as claimed.
\end{proof}

\begin{lemma}
    Let $X$ be the distance from $\pnt_1$ to the second closest point
    to it in $\PntSet \setminus \brc{\pnt_1}$.  Let $r_i =
    \sqrt{\frac{1}{i (m-1)\pi}}$.  Let $X_i$ be an indicator variable
    that is one if and only if $X \in [r_{i+1}, r_{i}]$.  For $i \geq
    4$, we have $\Prob{X_i=1} \leq 1/i^2$.
    
    \lemlab{too:close}
\end{lemma}

\begin{proof}
    In the following, let $p_i = \pi r_i^2 = 1/(i (m-1))$ be the
    probability of a random point to fall inside the disk of radius
    $r_i$ centered at $\pnt$. The probability that this disk has
    exactly $k$ points in it is exactly
    \begin{align*}
        \alpha_k = \binom{m-1}{k} p_i^k \pth{1-p_i}^{m-1-k}.
    \end{align*}
    These probabilities fall quickly with $k$, and can be treated as a
    geometric series. Indeed,
    \begin{align*}
        \frac{\alpha_{k+1}}{\alpha_k}%
        =%
        \frac{k!  (m-1-k)! p_i }{(k+1)!(m-1-k-1)! (1-p_i)}%
        =%
        \frac{(m-1-k)p_i}{(k+1)(1-p_i)}%
        \leq %
        \frac{2(m-1-k)}{i(m-1)(k+1)} \leq \frac{1}{2},
    \end{align*}
    for $i \geq 4$. Now, we have
    \begin{align*}
        \Prob{X \in [r_i, r_{i+1}] \MakeBig }%
        &\leq%
        \Prob{ \cardin{\disk(\pnt, r_i) \cap \pth{\PntSet \setminus
                 \brc{\pnt_1}}} \geq 2\MakeBig }%
        =%
        \sum_{k\geq 2} \alpha_k%
        \leq%
        2 \alpha_2\\
        &\leq 2 \frac{(m-1)(m-2)}{2} \cdot \frac{1}{i^2 (m-1)^2} \leq
        \frac{1}{i^2},
    \end{align*}
    as desired.
\end{proof}

\begin{lemma}
    If $X_i=1$ then there is a disk of radius $r_{i+1}/4$ contained
    inside $\VorCell{\pnt_1}{\PntSet}$.
\end{lemma}

\begin{proof}
    Let $\pntA$ and $\pntB$ be the closest and second closest nearest
    neighbor to $\pnt_1$, respectively, in $\PntSet \setminus
    \brc{\pnt_1}$.
    
    \parpic[r]{\includegraphics{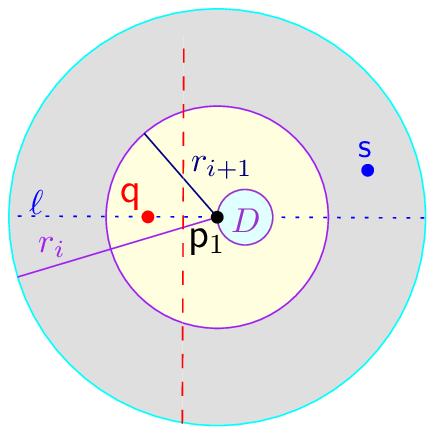}}
    
    Consider the line $\Line$ through $\pnt_1$ and $\pntA$, and place
    a point $\pntC$ on this line in distance $r_{i+1}/4$ from
    $\pnt_1$, on the side away from $\pntA$. We claim that the disk
    $D$ of radius $r_{i+1}/4$ centered at $\pntC$ is fully contained
    in $\VorCell{\pnt_1}{\PntSet}$. Indeed, observe that the Voronoi
    cells of all the points of $\PntSet$ in distance at least
    $r_{i+1}$ from $\pnt_1$ can not contain any point of $D$. As such,
    the only point that might be problematic is $\pntA$. But then,
    observe that the bisector between $\pnt_1$ and $\pntA$ is
    orthogonal to the line $\Line$, and it can not intersect $D$, thus
    implying the claim.
\end{proof}

The \emphi{fatness} of shape $X$ is $\fatX{X} = R(X)/ r(x)$, where
$R(x)$ (resp. $r(x)$) is the radius of the smallest (resp. largest)
disk enclosing (resp. enclosed in) $X$.

\begin{theorem}
    The fatness of $C= \VorCell{\pnt_1}{\PntSet}$ is constant, in
    expectation.
\end{theorem}
\begin{proof}
    Let $Y_j$ be an indicator variable that is one if and only if the
    diameter of $\VorCell{\pnt_1}{\PntSet}$ is in the range
    $\pbrc{R_{j-1},R_j}$. By \lemref{upper:bound}, we have that
    $\Prob{Y_j = 1} \leq 6 e^{1-j-1}$. Similarly, let $X_i$ be as in
    \lemref{too:close}. Note, that if $X_i=1$ and $Y_j=1$ then
    \begin{align*}
        F%
        =%
        \fatX{C} \leq \frac{R_j}{r_{i+1}/4}%
        \leq %
        \frac{4\cdot 3^{-1/4} \, \mysqrt{j/(m-1)}}%
        { \sqrt{4/ \pth{(i+1) (m-1)\pi}} } \leq %
        \frac{2\cdot 3^{-1/4} \, \mysqrt{j(i+1)}}%
        { \sqrt{1/ \pi} }%
        \leq 
        4 \sqrt{j(i+1)}.
    \end{align*}
    As such, we have $\Ex{ \fatX{C}} \leq \sum_{i=1}^\infty
    \sum_{j=1}^\infty \Prob{ \pth[]{X_i =1} \cap \pth[]{Y_j = 1}} 4
    \sqrt{j(i+1)}$. Setting $L(i) = 10 \ceil{\ln i} + 1$, we have that
    \begin{align*}
        \Ex{ \fatX{C}}%
        &\leq%
        \sum_{i=1}^\infty \sum_{j=1}^{L(i)} \Prob{ X_i =1} 4
        \sqrt{j(i+1)}%
        +%
        \sum_{i=1}^\infty \sum_{j=L(i) + 1}^\infty \Prob{ Y_j =1} 4
        \sqrt{j(i+1)}\\
        &\leq%
        \sum_{i=1}^\infty \frac{1}{i^2} 4 L(i) \sqrt{L(i)(i+1)}%
        +%
        \sum_{i=1}^\infty \sum_{j=L(i) + 1}^\infty 6 e^{1-j-1} \cdot 4
        \sqrt{j(i+1)}\\
        &=%
        \sum_{i=1}^\infty O\pth{ \frac{1}{i^{5/4}}} %
        +%
        \sum_{i=1}^\infty O\pth{\frac{1}{i^4}} = O(1).
    \end{align*}
\end{proof}

\end{document}

%% file: prefix.tex
\usepackage{amsmath}%
\usepackage{amssymb}%
\usepackage[breaklinks]{hyperref} 
\usepackage{graphicx}
\usepackage[active]{srcltx} 
\usepackage{euscript}%
\usepackage{xspace}%
\usepackage{color}%
\usepackage{amsmath}%
\usepackage{paralist}%
\usepackage{theorem}%
\usepackage{picins}%
\usepackage{boxedminipage}
\usepackage{enumerate}
\usepackage{picins}
\usepackage{multirow}

\ifx\SOCG\undefined
\newcommand{\InConfVer}[1]{}
\newcommand{\InFullVer}[1]{#1}
\else
\newcommand{\InConfVer}[1]{#1}
\newcommand{\InFullVer}[1]{}
\fi

\newcommand{\ts}{\hspace{0.6pt}}%
\renewcommand{\th}{\si{th}\xspace}

\providecommand{\si}[1]{#1}

\newcommand{\AnneThanks}[1]{\thanks{%
      Department of Information and Computing Sciences; Utrecht
      University; The Netherlands;
      \texttt{anne}\hspace{0cm}\texttt{\atgen{}cs.uu.nl}. %
      #1}}

\newcommand{\BenThanks}[1]{\thanks{Department of Computer Science;
      University of Illinois; %
      201 N. Goodwin Avenue; %
      Urbana, IL, 61801, USA; %
      {\tt \si{raichel}2\atgen{}\si{uiuc}.\si{edu}}; %
      {\tt \url{\si{http://www.cs.uiuc.edu/\string~\si{raichel2}}}}%
      . %
      #1}}

\newcommand{\atgen}{\symbol{'100}}
\newcommand{\SarielThanks}[1]{\thanks{Department of Computer Science;
      University of Illinois; 201 N. Goodwin Avenue; Urbana, IL,
      61801, USA; {\tt \si{sariel}\atgen{}\si{uiuc.edu}}; %
      {\tt \url{http://www.uiuc.edu/\string~sariel/}.} #1}}

\definecolor{blue25}{rgb}{0,0,0.55}%
\newcommand{\emphic}[2]{%
   \textcolor{blue25}{%
      \textbf{\emph{#1}}}%
   \index{#2}}

\newcommand{\emphi}[1]{\emphic{#1}{#1}}
\newcommand{\pth}[2][\!]{#1\left({#2}\right)}

\definecolor{red25}{rgb}{0.4,0,0.0}%

\newtheorem{theorem}{Theorem}[section]
\newtheorem{lemma}[theorem]{Lemma}%
\newtheorem{definition}[theorem]{Definition}

\newtheorem{corollary}[theorem]{Corollary}
\newtheorem{fact}[theorem]{Fact}%

{\theorembodyfont{\rm} \newtheorem{observation}[theorem]{Observation}}
{\theorembodyfont{\rm} }
   {\theorembodyfont{\rm} \newtheorem{example}[theorem]{Example}}

\newcommand{\brc}[1]{\left\{ {#1} \right\}}

\newcommand{\west}{west\xspace}
\newcommand{\east}{east\xspace}
\newcommand{\south}{south\xspace}
\newcommand{\north}{north\xspace}
\newcommand{\southeast}{south-east\xspace}
\newcommand{\northwest}{north-west\xspace}

\newcommand{\Disk}{\mathsf{d}}%

\newcommand{\Term}[1]{\textsf{#1}}%

\newcommand{\PVD}{\Term{PVD}\xspace}

\newcommand{\ringX}[1]{\EuScript{R}_{#1}}

\renewcommand{\Re}{{\rm I\!\hspace{-0.025em} R}}

\newcommand{\aftermathA}{\par\vspace{-\baselineskip}}%

\newcommand{\MakeBig} {\rule[-.2cm]{0cm}{0.4cm}}
\newcommand{\MakeSBig}{\rule[0.0cm]{0.0cm}{0.38cm}} 

\newcommand{\seclab}[1]{\label{sec:#1}}
\newcommand{\secref}[1]{Section~\ref{sec:#1}}

\newcommand{\thmlab}[1]{{\label{theo:#1}}}%
\newcommand{\thmref}[1]{Theorem~\ref{theo:#1}}%

\providecommand{\deflab}[1]{\label{def:#1}}
\newcommand{\defref}[1]{Definition~\ref{def:#1}}

\newcommand{\corlab}[1]{\label{cor:#1}}
\newcommand{\corref}[1]{Corollary~\ref{cor:#1}}

\newcommand{\exelab}[1]{\label{exer:#1}}
\newcommand{\exeref}[1]{Exercise~\ref{exer:#1}}

\newcommand{\lemlab}[1]{\label{lemma:#1}}
\newcommand{\lemref}[1]{Lemma~\ref{lemma:#1}}

\newcommand{\obslab}[1]{\label{observation:#1}}
\newcommand{\obsref}[1]{Observation~\ref{observation:#1}}

\newcommand{\factlab}[1]{\label{fact:#1}}
\newcommand{\factref}[1]{Fact~\ref{fact:#1}}

\newenvironment{proof}{\trivlist\item[]\emph{Proof}:}%
{\unskip\nobreak\hskip 1em plus 1fil\nobreak%
   \myqedsymbol
   \parfillskip=0pt%
   \endtrivlist}
{\unskip\nobreak\hskip 1em plus 1fil\nobreak%
   \myqedsymbol
   \parfillskip=0pt%
   \endtrivlist}

\newcommand{\myqedsymbol}{\rule{2mm}{2mm}}
\newcommand{\cardin}[1]{\left| {#1} \right|}
\newcommand{\pbrcx}[1]{\left[ {#1} \right]}
\newcommand{\Ex}[2][\!]{\mathop{\mathbf{E}}#1\pbrcx{#2}}
\newcommand{\Prob}[1]{\mathop{\mathbf{Pr}}\!\pbrcx{#1}}
\newcommand{\sep}[1]{\,\left|\, {#1} \MakeBig\right.}

\providecommand{\ds}{\displaystyle}

\newcommand{\remove}[1]{}

\newcommand{\ceil}[1]{\left\lceil {#1} \right\rceil}
\newcommand{\floor}[1]{\left\lfloor {#1} \right\rfloor}

\newcommand{\etal}{\textit{et~al.}\xspace}

\usepackage{pifont}%

\usepackage[symbol*]{footmisc}

\DefineFNsymbols*{stars}[text]{%
   {\ding{172}} %
   {\ding{173}} %
   {\ding{174}} %
   {\ding{175}} %
   {\ding{176}} %
   {\ding{177}} %
   {\ding{178}} %
   {\ding{179}} %
   {\ding{180}} %
   }
\setfnsymbol{stars}

\newcommand{\site}{\mathsf{s}}
\newcommand{\siteA}{\mathsf{t}}
\newcommand{\siteB}{\mathsf{r}}

\newcommand{\pnt}{{\mathsf{p}}}%
\newcommand{\pntA}{{\mathsf{q}}}%
\newcommand{\pntB}{{\mathsf{s}}}
\newcommand{\pntC}{{\mathsf{t}}}

\newcommand{\Line}{\ell}

\newcommand{\PntSet}{\mathsf{P}}%
\newcommand{\Terrain}{\EuScript{T}}
\newcommand{\Triangulation}{\Delta}
\newcommand{\Vertices}{\mathsf{V}}
\newcommand{\Edges}{\mathsf{E}}
\newcommand{\Domain}{D}
\newcommand{\Sites}{\mathsf{P}}
\newcommand{\VD}[1]{\EuScript{V}or(#1)}
\newcommand{\ldConst}{\lambda}
\newcommand{\slConst}{\xi}

\newcommand{\distGeo}[2]{\mathsf{d}_{\Terrain}\pth{#1,#2}}
\newcommand{\distX}[2]{\|#1-#2\|}

\newcommand{\distSet}[2]{\mathsf{d}\pth{#1, #2}}

\newcommand{\Cell}{\mathsf{C}}
\newcommand{\Chords}{\mathsf{B}}
\newcommand{\Intersections}{\chi}
\newcommand{\Graph}{\EuScript{G}}
\newcommand{\chords}{chords\xspace}

\newcommand{\slFactor}{\beta}
\newcommand{\diamConst}{\delta}

\newcommand{\constA}{c_1}%
\newcommand{\constB}{c_2}%
\newcommand{\constC}{c_3}%
\newcommand{\VorCell}[2]{\mathrm{cell}\pth{#1, #2}}
\newcommand{\mysqrt}[1]{\ts\!\!\sqrt{#1}}
\newcommand{\fatX}[1]{\mathrm{fat}\pth{#1}}
\newcommand{\pbrc}[2][\!\!]{#1\left[ {#2} \MakeBig \right]}

\newcommand{\Square}{\Box}
\newcommand{\Sample}{\mathsf{P}}%
\newcommand{\disk}{\mathrm{disk}}%

\newcommand{\Renyi}{R{\'e}nyi\xspace}

\newcommand{\apndlab}[1]{\label{apnd:#1}}
\newcommand{\apndref}[1]{Appendix~\ref{apnd:#1}}
\newcommand{\slTFactor}{\slFactor^5}
